\documentclass[10pt,journal,doublecolumn]{IEEEtran}

\usepackage{threeparttable,placeins,makecell,stfloats,cite}
\usepackage{threeparttable,arydshln}
\usepackage{amsmath,amssymb}
\usepackage{amsthm}
\usepackage[dvips]{graphicx}
\usepackage[usenames]{color}
\usepackage{amsfonts}
\usepackage{latexsym}
\usepackage{multirow}
\usepackage{caption}
\usepackage{rotating}
\usepackage{float}
\usepackage[format=hang]{subfig}
\usepackage{url}
\usepackage{cite}
\usepackage{microtype}

\newtheorem{theorem}{\textit{Theorem}}
\newtheorem{lemma}{\textit{Lemma}}

\newtheorem{proposition}{\textit{Proposition}}

\newtheorem{remark}{\textit{Remark}}

\newtheorem{construction}{\textit{Construction}}

\begin{document}
	
	%\title{Dense Code Multiple Access}
	\title{An Explicit Upper Bound of Generalized Quadratic Gauss Sums and Its Applications for Asymptotically Optimal Aperiodic Polyphase Sequence Design}
	
	\author{Huaning~Liu~and~Zilong~Liu%,~\textit{IEEE Senior Member}
		\thanks{Huaning Liu is with the Research Center for Number Theory and Its Applications, School of Mathematics, Northwest University, Xi'an 710127, China 
			(e-mail: hnliu@nwu.edu.cn). 	
			Zilong Liu is with the School of Computer Science and Electronics Engineering, University of Essex, UK (e-mail: zilong.liu@essex.ac.uk). 
            \textit{This work has been submitted to IEEE Transactions on Information Theory on 22 January 2026}. }
	}
	
	\maketitle
	
	\begin{abstract}
	This work is motivated by the long-standing open problem of designing \textit{asymptotically} order-optimal aperiodic polyphase sequence sets with respect to the celebrated Welch bound. Attempts were made by Mow over 30 years ago, but a comprehensive understanding to this problem is lacking. Our first key contribution is an explicit upper bound of generalized quadratic Gauss sums which is obtained by recursively applying Paris' asymptotic expansion and then bounding it by leveraging the fast convergence property of the Fibonacci zeta function. Building upon this major finding, our second key contribution includes four systematic constructions of order-optimal sequence sets with low aperiodic correlation and/or ambiguity properties via carefully selected Chu sequences and Alltop sequences. For the first time in the literature, we reveal that the full Alltop sequence set is asymptotically optimal for its low aperiodic correlation sidelobes. Besides, we introduce a novel subset of Alltop sequences possessing both order-optimal aperiodic correlation and ambiguity properties for the entire time-shift window. 

    %challenge of sequence design for communication and radar systems, specifically focusing on the lack of tight explicit bounds for . We propose a theoretical breakthrough by deriving the first explicit upper bound for these sums, with a constant factor that is directly applicable to sequence performance analysis. Additionally, we present an engineering breakthrough with the construction of four order-optimal sequence sets. Among these, the Alltop-derived set $A_2$ achieves a novel dual-optimality, balancing both low aperiodic correlation and low ambiguity, thus resolving the long-standing trade-off in sequence design. Numerical results demonstrate that the $A_2$ set outperforms existing sequence families in both correlation and ambiguity, making it highly suitable for integrated communication-radar systems. Our contributions provide a foundation for provably optimal sequence design, enhancing the performance of next-generation dynamic wireless systems.
	\end{abstract}
	
	\begin{IEEEkeywords}
		Generalized quadratic Gauss sums, polyphase sequences, order-optimal sequence sets, aperiodic correlation, Welch bound, ambiguity function (AF), integrated sensing and communication (ISAC), high-mobility channels, Doppler effect.
	\end{IEEEkeywords}
	
	\section{Introduction}
    \subsection{Background}
	\IEEEPARstart{S}equences with low aperiodic correlation and/or low ambiguity properties play an instrumental role in radar sensing, wireless communications, as well as the integration of these two, called integrated sensing and communications (ISAC), which is a major use case in the upcoming sixth generation (6G) systems and beyond \cite{Fan-book, Golomb-book, Sturm2011, Liu2018}. While low aperiodic-correlation properties are useful for combating intersymbol interference and co-channel (or multi-user) interference, low aperiodic-ambiguity properties are desirable for dealing with the Doppler effect which is widely present in high mobility channels. The optimal design of such sequence sets is essential for achieving enhanced wireless system performances in terms of the sensing/positioning accuracy, transmission reliability, anti-interference/jamming capability, spectral efficiency, etc. Compared to their periodic counterparts \cite{Tor-book}, however, bounding the maximum aperiodic-correlation (or aperiodic-ambiguity) magnitude of a sequence set is less tractable. This is mainly because of the paucity of analytical mathematical tools that can efficiently deal with the calculation of aperiodic correlation (or ambiguity) values over incomplete periods of time- and/or Doppler-shifts.     
    %n modern communication, radar, and integrated sensing-communication systems, sequence design is a critical component determining overall system performance. These systems, especially high-speed communication and high-resolution radar, impose stringent yet conflicting requirements on polyphase sequence sets: low aperiodic correlation to mitigate intersymbol and co-channel interference, and low ambiguity function to counteract Doppler shifts in dynamic environments. These performance metrics directly affect signal detection accuracy, anti-interference capability, and spectral efficiency, making them essential for the success of next-generation wireless systems.
	
	%The definitions of the aperiodic correlation function and ambiguity function are as follows.
	Let $\mathbf{a}=(a_0, a_1, \cdots, a_{L-1})$ and $\mathbf{b}=(b_0, b_1, \cdots, b_{L-1})$ be two complex sequences of length $L$. The aperiodic correlation function $R_{\mathbf{a}, \mathbf{b}}(\tau)$
	of sequences $\mathbf{a}$ and $\mathbf{b}$ is defined by
	$$
	R_{\mathbf{a}, \mathbf{b}}(\tau):=\left\{
	\begin{array}{ll}
		\displaystyle\sum_{t=0}^{L-1-\tau}a_{t}(b_{t+\tau})^{*}, & 0\leq \tau\leq L-1, \\
		\displaystyle\sum_{t=0}^{L-1+\tau}a_{t-\tau}(b_{t})^{*}, & -(L-1)\leq \tau\leq -1, \\
		0, & |\tau|\geq L,
	\end{array}
	\right.
	$$ 
	where $(\cdot)^{*}$ denotes the complex conjugate of $(\cdot)$. When $\mathbf{a}=\mathbf{b}$, it is called the aperiodic auto-correlation function; otherwise, it is called the aperiodic cross-correlation function. Clearly, 
	$R_{\mathbf{a}, \mathbf{b}}(\tau)=R^*_{\mathbf{b}, \mathbf{a}}(-\tau)$ when $\tau<0$.
	
	Let $\mathcal{S}$ be a sequence set with $|\mathcal{S}|=K$. The aperiodic auto-correlation tolerance $\delta_{a}(\mathcal{S})$ and the aperiodic cross-correlation tolerance $\delta_{c}(\mathcal{S})$ of $\mathcal{S}$ are respectively defined as 
	\begin{eqnarray*}
		\delta_{a}(\mathcal{S})&:=&\max\left\{|R_{\mathbf{a}, \mathbf{a}}(\tau)|: \ \mathbf{a}\in\mathcal{S}, \  0<\tau\leq L-1\right\}, \\
		\delta_{c}(\mathcal{S})&:=&\max\left\{|R_{\mathbf{a}, \mathbf{b}}(\tau)|: \ \mathbf{a}, \mathbf{b}\in\mathcal{S}, \ \mathbf{a}\neq \mathbf{b}, \right. \\
		&&\qquad\qquad\qquad\qquad\left. 0\leq \tau\leq L-1\right\}.
	\end{eqnarray*}	
	The ``aperiodic tolerance" (also called the ``maximum aperiodic-correlation magnitude") $\delta_{\max}(\mathcal{S})$ of $\mathcal{S}$ is defined as
	\begin{eqnarray*}
		\delta_{\max}(\mathcal{S}):=\max\left\{\delta_{a}(\mathcal{S}), \ \delta_{c}(\mathcal{S})\right\}.
	\end{eqnarray*}	
	Welch \cite{Welch1974} showed that 
	\begin{eqnarray}\label{Equation:Welch bound}
		\delta_{\max}(\mathcal{S})\geq L \cdot \sqrt{\frac{K-1}{K(2L-1)-1}}.
	\end{eqnarray}	
	
	In 1980, McEliece first considered the following problem \cite{McEliece1980}:
    
	\textit{Given any fixed $K$, what is the order of $\delta_{\max}(\mathcal{S})$ in terms of $L$? }
    
    A set of binary sequences derived from the irreducible cyclic codes was introduced, but with an aperiodic tolerance of $\mathcal{O}(\sqrt{L}\log L)$. 
	
	Clearly, the Welch bound in (\ref{Equation:Welch bound}) suggests that $\delta_{\max}(\mathcal{S})$ should be at least $\mathcal{O}(\sqrt{L})$. $\mathcal{S}$ is said to be an asymptotically \textit{order-optimal} aperiodic-correlation sequence set if it has $\delta_{\max}(\mathcal{S})=\mathcal{O}(\sqrt{L})$ in this work. 
	%if its maximum aperiodic correlation magnitude $\delta_{\text{max}}(S) \ll \sqrt{L}$, matching the Welch lower bound  for fixed $K$. 
    The design of order-optimal aperiodic polyphase sequence sets was studied in Mow's PhD thesis over 30 years ago. By choosing a primary Chu sequence and its conjugate, it was proved that the aperiodic tolerance of the corresponding sequence pair is less than $1.487\sqrt{L}$ for even $L$ and $1.122\sqrt{L}$ for odd $L$ \cite{Mow1993}. Nonetheless, when $K>2$, no progress has been reported so far. Although another sequence family was shown to be order-optimal for $K\geq 2$ in \cite{Mow1993,Mow1994}, that design heavily relies on zero padding and with varying maximum magnitudes for different sequences. 
    
    %While Mow \cite{Mow1994} introduced the first order-optimal set, it suffered from excessive zero elements. %For polyphase sequences, only a single order-optimal sequence pair (primary Chu sequence and its conjugate \cite{Mow1993}) existed for 30 years.	

	%Let $\mathbf{a}=(a_0, a_1, \cdots, a_{L-1})$ and $\mathbf{b}=(b_0, b_1, \cdots, b_{L-1})$ be two complex sequences of length $L$. 
    Next, with the aid of $\mathcal{S}$, let us introduce order-optimal sequences with low aperiodic ambiguity properties. Let $Z_x, Z_y\in[1, L]$ be real numbers (often positive integers) and define the low ambiguity zone (LAZ) of interest as follows: 
	$$
	\Pi=\left\{(\tau, \nu)\in \mathbb{Z}\times \mathbb{Z}: \ \tau\in(-Z_x, Z_x), \nu\in(-Z_y, Z_y)\right\}.
	$$
	Define $e(x):={e}^{2\pi i x}$, where $i=\sqrt{-1}$.
	The aperiodic ambiguity function $A_{\mathbf{a}, \mathbf{b}}(\tau, \nu)$
	of sequences $\mathbf{a}$ and $\mathbf{b}$ at delay shift $\tau$ and Doppler shift $\nu$ 
	is defined as
	\begin{eqnarray*}
		&&A_{\mathbf{a}, \mathbf{b}}(\tau, \nu)\\
		&&:=\left\{
		\begin{array}{ll}
			\displaystyle\sum_{t=0}^{L-1-\tau}a_{t}(b_{t+\tau})^{*}e\left(\frac{\nu t}{L}\right), & 0\leq \tau\leq Z_x, \\
			\displaystyle\sum_{t=0}^{L-1+\tau}a_{t-\tau}(b_{t})^{*}e\left(\frac{\nu t}{L}\right), & -Z_x\leq \tau\leq -1, \\
			0, & |\tau|\geq Z_x,
		\end{array}
		\right.
	\end{eqnarray*}	
	where $(\tau, \nu)\in \Pi$. Further, define
	\begin{eqnarray*}
		\theta_{a}(\mathcal{S})&:=&\max\left\{|A_{\mathbf{a}, \mathbf{a}}(\tau, \nu)|: \ \mathbf{a}\in\mathcal{S}, \  (\tau, \nu)\in \Pi \right. \\
		&&\qquad\left.\text{ \ with \ } (\tau, \nu)\neq (0, 0)\right\}, \\
		\theta_{c}(\mathcal{S})&:=&\max\left\{|A_{\mathbf{a}, \mathbf{b}}(\tau, \nu)|: \ \mathbf{a}, \mathbf{b}\in\mathcal{S}, \ \mathbf{a}\neq \mathbf{b},  \right.\\
		&&\qquad\left.(\tau, \nu)\in \Pi\right\}, \\
		\theta_{\max}(\mathcal{S})&:=&\max\left\{\theta_{a}(\mathcal{S}), \ \theta_{c}(\mathcal{S})\right\},
	\end{eqnarray*}	
    whereby $\theta_{\max}(\mathcal{S})$ is called the ambiguity tolerance of $\mathcal{S}$ associated to the delay-Doppler region of $\Pi$. 

    It was pointed out in \cite{MengGGLF2025} that $\theta_{\max}(\mathcal{S})$ should also be at least $\mathcal{O}(\sqrt{L})$ provided that $\mathcal{O}(Z_xZ_y)\geq \mathcal{O}(L)$. Under this condition, $\mathcal{S}$ is said to be an asymptotically \textit{order-optimal} aperiodic-ambiguity sequence set if $\theta_{\max}(\mathcal{S})=\mathcal{O}(\sqrt{L})$ is met. However, characterizing the maximum aperiodic-ambiguity magnitude of a sequence set is more challenging due to its two-dimensional function nature involving both time- and Doppler- shifts. The only known construction of order-optimal aperiodic-ambiguity sequence sets was proposed in \cite{MengGGLF2025}, whereby Van Der Corput's technique was adopted for bounding the exponential sum associated to a twice continuously differentiable function. By carefully selecting certain co-located Chu sequences (in terms of their root indices), it was shown that $\theta_{\max}(\mathcal{S})=\mathcal{O}(\sqrt{KL})$ for relatively small $K$. It is noted that although there is a rich body of literature on the bounds and optimal constructions of sequences with low periodic correlation/ambiguity properties (see \cite{BenedettoD2007,DingFFXZ2013,Gong2013,GuZFAL2025,
		Schmidt2011,ShenYZLF2025,TianSLL2025,WangG2011,WangGY2013,WangSYZ2025,YeZFLLT2022}), those results may not be applicable to our aperiodic design problems.  
	%Sequences with low ambiguity functions are crucial for mitigating the Doppler effect in advanced communication and radar systems, enabling accurate signal detection and processing in dynamic environments.
	%However, analyzing the ambiguity function poses a significantly greater challenge than assessing traditional one-dimensional correlation. Its two-dimensional nature and inherent complexity make it notoriously difficult to determine exact values, leading to a dearth of comprehensive analytical results. Many papers have been written on the bounds of $\theta_{\max}(\mathcal{S})$, and the constructions of sequences with low ambiguity functions
	%(see \cite{BenedettoD2007,DingFFXZ2013,Gong2013,GuZFAL2025,MengGGLF2025,
		%Schmidt2011,ShenYZLF2025,TianSLL2025,WangG2011,WangGY2013,WangSYZ2025,YeZFLLT2022}). A sequence set $\mathcal{S}$ with $\theta_{\max}(\mathcal{S})\ll\sqrt{L}$ is said to be “order-optimal”. 
	%Some heuristic discussions on the order-optimal sequence sets were provided in \cite{MengGGLF2025}.

\subsection{Motivations and Contributions}
%While numerous studies have addressed the problem of sequence design, most existing sequences either excel in one metric at the expense of the other or fail to provide quantifiable performance bounds. Chu sequences, known for their low correlation, suffer from poor ambiguity control, whereas Alltop sequences, with their cubic-term structure, offer resilience to Doppler effects but lack comprehensive performance characterization. The need for sequences that simultaneously achieve low correlation and low ambiguity, without compromising either, has thus remained an unsolved challenge.
When Chu sequences are employed as the major building component for such order-optimal aperiodic sequence design, a fundamental research task is to explicitly bound (besides bounding it up to certain order to magnitude) the generalized quadratic Gauss sums\footnote{Strictly speaking, we are concerned with \textit{generalized incomplete quadratic Gauss sums} in this work. For ease of presentation, ``incomplete" is dropped but it is underlying.} which is a mathematical branch deeply rooted in advanced number theory. From this standpoint, Van Der Corput's technique may not be able to yield a tight bound since its exponential sum does not directly deal with a quadratic function. 

Formally, for real $x$, $\theta$ and positive integer $N$, the generalized quadratic Gauss sum is defined as
	$$
	S_{N}(x, \theta):=\sum_{t=1}^{N}e\left(\frac{xt^2}{2}+\theta t\right),
	\quad 0<x<1, \ -\frac{1}{2}<\theta\leq\frac{1}{2}.
	$$
	%where $e(x)=\hbox{e}^{2\pi i x}$, $i=\sqrt{-1}$. 
	Gauss first established foundational results for complete quadratic Gauss sums and proved that
	\begin{eqnarray*}
		\left|S_{q}\left(\frac{2a}{q}, \frac{b}{q}\right)\right|
		=\left\{\begin{array}{ll}
			\displaystyle \sqrt{q}, & \text{ if $q$ is odd}, \\
			\displaystyle \sqrt{\frac{q}{2}}\left(1+(-1)^{a\frac{q}{2}+b}\right), & \text{ if $q$ is even},
		\end{array}
		\right.
	\end{eqnarray*}
	where $q, a, b$ are integers with $q\geq 1$ and $\gcd(a, q)=1$. 
	Based on the approximate function relation
	\begin{eqnarray}\label{Equation:approximate function relation}
		S_{N}(x, \theta)&=&\frac{e^{-\pi i\theta^2/x+\pi i/4}}{\sqrt{x}}S_{\lfloor Nx\rfloor}\left(-\frac{1}{x}, \frac{\theta}{x}\right)\nonumber \\
		&&+\mathcal{O}\left(\frac{1+|\theta|}{\sqrt{x}}\right),
	\end{eqnarray}
	the sum $S_{N}(x, \theta)$ over $N$ terms can be approximated
	by a similar sum taken over $\lfloor Nx\rfloor$ terms. 
	%where we write $f(n)=\mathcal{O}(g(n))$ or $f(n)\ll g(n)$ if $|f(n)|\le c g(n)$ for some absolute constant $c>0$. 
    
    Repeated application of (\ref{Equation:approximate function relation}) enables the
	representation of $S_{N}(x, \theta)$ as a steadily decreasing number of error terms. Using this method, Hardy and Littlewood \cite{HardyL1914} showed that $S_{N}(x, \theta)={o}(N)$ for any irrational $x$, with more precise
	order estimates depending on the detailed arithmetic structure of $x$. Littlewood \cite{Littlewood1961} proved that
	\begin{eqnarray}\label{Equation:Chu 1}
		\left|S_N\left(\frac{1}{N}, \theta\right)\right|\leq 1.35\sqrt{N},~\text{for odd}~N. 
	\end{eqnarray}
	Subsequent works by Lehmer \cite{Lehmer1976}, Fiedler, Jurkat and K\"{o}rner \cite{FiedlerJK1977}, 
    Sullivan and Zannier \cite{SullivanZ1992} and Korolev \cite{Korolev2015} provided asymptotic estimates or upper bounds in special circumstances. But these results are only applicable to certain parameter settings and fail to cope with more complex quadratic Gauss sums. 
    
    In 2014, Paris \cite{Paris2014} gave an asymptotic expansion for the generalized quadratic Gauss sum $S_{N}(x, \theta)$ and expressed the error term in
		(\ref{Equation:approximate function relation}) in terms of complementary error functions $E(x, \theta)$ (see \textit{Proposition 1} for its formal definition). Nevertheless, straightforward application of such an asymptotic expansion for our bounding task is not obvious. Although G\"{u}nther and Schmidt applied Paris' technique to analyze the collective smallness of the aperiodic auto-correlations and the aperiodic cross-correlations of certain Chu sequence pairs \cite{GuntherS2019}, their approach may not serve the purpose of this work as we are concerned with the peak aperiodic sidelobes of sequences. % needed for the search of order-optimal aperiodic sequences.
	
%Previous work in this area has provided asymptotic results, but these are insufficient for practical engineering applications, as they do not offer exact performance guarantees. 
To address the aforementioned research gaps, we present an explicit bound for generalized quadratic Gauss sums, thus paving the way for efficient search of order-optimal aperiodic sequences. Two main contributions of this work are summarized as follows: 
	
	\begin{itemize}
		
		\item Based on Paris' asymptotic expansion technique, we propose a scenario-specific error control strategy for the complementary error functions $E(x, \theta)$. By distinguishing between cases of $|\theta|>{\sqrt{x}}/{\sqrt{\pi}}$ and $|\theta|\leq{\sqrt{x}}/{\sqrt{\pi}}$, respectively, we are able to determine the implied constant of the error term in (\ref{Equation:approximate function relation}). Similar to the method of Hardy and Littlewood in \cite{HardyL1914}, one can recursively use Paris' asymptotic expansion to express $S_{N}({a}/{q}, \theta)$ as a steadily decreasing number of error terms, but bounding the sum of all these error terms for such a general case is non-trivial. Interestingly, it is found that these decreasing terms can be controlled by a Fibonacci zeta function and this key observation permits us to show that $\left|S_{N}\left({a}/{q}, \theta\right)\right|< 20.07\sqrt{q}+3$.  
		
		\item Building upon the derived bound, we introduce four families of order-optimal polyphase sequence sets: two induced by decimated Chu sequences \cite{Chu1972} where $\mathcal{C}_1$ and $\mathcal{C}_2$ emphasizing low aperiodic correlation/ambiguity properties, respectively, and another two induced by Alltop sequences \cite{Alltop1980}, i.e., $\mathcal{A}_1$ and $\mathcal{A}_2$. Interestingly, $\mathcal{C}_1$ reduces to Mow's order-optimal Chu sequence pair \cite{Mow1993} when $K=2$. Albeit Alltop sequences are long known to be optimal with respect to Welch's periodic cross-correlation bound, we reveal, for the first time, that the full set of the Alltop sequences (i.e., $\mathcal{A}_1$) is asymptotically optimal for its low aperiodic correlation sidelobes. Furthermore, it is shown that $\mathcal{A}_2$ (i.e., a subset of $\mathcal{A}_1$) achieves simultaneous asymptotic optimality in terms of both low aperiodic correlation and ambiguity properties for the entire time-shift window. 
        %Specifically, we have 
		%$$		
		%\delta_{\max}\left(\mathcal{A}_{2}\right) \leq 21 \sqrt{p}, %\quad\theta_{\max}\left(\mathcal{A}_{2}\right) \leq 21 \sqrt{p},
		%$$
		%where $p$ refers to the prime sequence length.
	\end{itemize}
	
	The remainder of this paper is structured as follows: Section II reviews foundational theories and key lemmas, including refined error estimation for quadratic Gauss sum asymptotic expansion. Section III presents the  proof of the proposed explicit bound on generalized quadratic Gauss sum. Sections IV and V detail the proposed constructions of order-optimal sequence sets based on Chu and Alltop sequences, respectively. Section VI concludes our key contributions and outlines future research directions.

	\vspace{0.2in}
	\section{PRELIMINARY}
	
	This section briefly introduces Chu and Alltop sequences, incomplete exponential sums, our refined error bound for Paris' asymptotic expansion, as well as some basics on the Fibonacci zeta function for enhanced bounding of the recursive error accumulation.
	
	\subsection{Chu and Alltop sequences}
	%Define $e(x)=\hbox{e}^{2\pi ix}$, where $i=\sqrt{-1}$, and 
    Let $N$ be a positive integer. The Chu sequences \cite{Chu1972} are defined by
	\begin{eqnarray*}
		&&\mathbf{c}^{r}:=\left(c_0^{r}, c_1^{r}, \cdots, c_{N-1}^{r}
		\right), \\
		&&c_t^{r}:=
		\left\{\begin{array}{ll}
			\displaystyle 	e\left(\frac{rt(t-1)}{2N}\right), & N \text{ \ is odd}, \\
			\displaystyle 	e\left(\frac{rt^2}{2N}\right), &  N \text{ \ is even}, \\
		\end{array}
		\right. \\
		&&\quad 0\leq t\leq N-1, \ 1\leq r\leq N-1.
	\end{eqnarray*}
	In particular, $\mathbf{c}^{1}$ is called the primary Chu sequence with its conjugate $\mathbf{c}^{L-1}$.
	Many papers are concerned with the aperiodic auto-correlation function
	and cross-correlation function of Chu sequences (see \cite{FanDH1994,GabidulinFD1995,GuntherS2019,KangWKK2012,LeeY2014,Mercer2013,MowL1997,ZhangG1993} for details), but their aperiodic ambiguity properties remain largely unexplored.
	
	For prime $p\geq 5$ the Alltop sequences \cite{Alltop1980} are defined as
	\begin{eqnarray*}
		&&\mathbf{a}^{r}:=\left(a_0^{r}, a_1^{r}, \cdots, a_{p-1}^{r}
		\right), \\
		&&a_t^{r}:=e\left(\frac{t^3+rt}{p}\right),
		\quad 0\leq t\leq p-1, \ 0\leq r\leq p-1.
	\end{eqnarray*}
	The cubic exponential phase term of Alltop sequences introduces inherent Doppler resilience, yet explicit performance bounds remain elusive.
	
	\subsection{Existing Results on Incomplete Exponential Sum}
	
	For integers $N$, $r$ and $\tau$ with $1\leq r\leq N-1$ and $1\leq\tau\leq N-1$ we define
	$$
	T_N(r, \tau)=\sum_{t=0}^{N-1-\tau}e\left(-\frac{r\tau t}{N}\right).
	$$
	Clearly, 
	$$
	|T_N(r, \tau)|=
	\left|\frac{\sin\frac{\pi r\tau^2}{N}}{\sin\frac{\pi r\tau}{N}}\right|.
	$$
	
	Define
	$$
	B_{r}:=\max_{\tau=1,2,\cdots,N-1}|T_N(r, \tau)|:=\left|T_N(r, I_m(N))\right|,
	$$
	where $I_{m}(N)$ is the value of $\tau$ $(1\leq\tau\leq L-1)$ which maximizes
	$|T_N(r, \tau)|$.
	Zhang and Golomb \cite{ZhangG1993} showed that
	$$
	B_1=B_{N-1}=\sqrt{N/4.34}.
	$$
	Fan, Darnell and Honary \cite{FanDH1994} proved that
	$$
	B_{\frac{N\pm 1}{2}}=\sqrt{N/2.174}.
	$$
	Gabidulin, Fan and Darnell \cite{GabidulinFD1995} studied the asymptotic behavior of $B_{r}$.
	
	\begin{lemma}\label{Lemma:estimate-of-exponential-sum-1} 
		For $r\geq 2$ we have
		\begin{eqnarray*}
			B_{r}\approx \left\{\begin{array}{ll}
				\displaystyle 0.48\sqrt{\frac{b}{r}}N, & I_m(N)=\frac{(bN-1)s_0}{r}, \ 0\leq\frac{b}{r}\leq 0.37, \\	
				\displaystyle \frac{N}{\pi}\sin\left(\frac{\pi b}{r}\right), &  I_m(N)=\frac{bN-1}{r}, \ 0.37\leq \frac{b}{r}\leq 0.5, 
			\end{array}
			\right.
		\end{eqnarray*}
		where $bN\equiv \pm 1 \ (\bmod \ r)$, $1\leq b\leq \lfloor\frac{r}{2}\rfloor$, $s_0=\sqrt{\frac{z_0r}{\pi b}}$ and
		$z_0=1.1655$.	
	\end{lemma}

	\subsection{A Refined Bound on Paris' Asymptotic Expansion Errors}
	
	Paris' asymptotic expansion for the generalized quadratic Gauss sum
	$S_{N}(x, \theta)$ was given in \cite{Paris2014}. 
	
	\begin{proposition}\label{Proposition: expression} Let $N$ be a positive integer, $x\in (0,1)$, $\theta\in(-1/2, \ 1/2]$. Write
		$Nx+\theta=M+\epsilon$, where $M$ is an integer and $\epsilon\in(-1/2, \ 1/2]$. Then
		\begin{eqnarray*}
			S_N(x,\theta)&=&\frac{e^{-\pi i\theta^2/x+\pi i/4}}{\sqrt{x}}S_M(-1/x,\theta/x)+\frac{\mu - 1}{2} \\
			&&+\frac{e^{\pi i/4}}{2\sqrt{x}}\big(E(x,\theta)-\mu E(x,\epsilon)\big)\\
			&&+\frac{i}{2}\big(g(\theta)-\mu g(\epsilon)\big)+R,
		\end{eqnarray*}
		where $|R|<x$, $\mu = e^{\pi i xN^2 + 2\pi i\theta N}$,  and $g: \ [-1/2, \ 1/2]\to\mathbb{R}$ is defined by
		$$
		g(t)=\left\{\begin{array}{ll}
			0, &\text{for $t = 0$}\\
			\cot(\pi t)-(\pi t)^{-1}, &\text{otherwise},
		\end{array}\right.
		$$
		and $E(x, \theta)=e^{-\pi i\theta^2/x}\textup{erfc}(e^{\pi i/4}\theta\sqrt{\pi/x})$, where 
		$\textup{erfc}(z)=1 - \frac{2}{\sqrt{\pi}}\int_{0}^{z}e^{-t^2}dt$. 
	\end{proposition}
	
	%\textit{Proof}: 
    \begin{proof}
    See Theorem 1 of \cite{Paris2014} or Proposition 5 of \cite{GuntherS2019}.
    \end{proof}
	
	The error term in Paris' asymptotic expansion for generalized quadratic Gauss sums is not explicit. Next, we refine this expansion to obtain a critical lemma for quantifiable error control. 	
	
	From (1.3) and (1.4) of \cite{Paris2014}, we have the asymptotic expansion of $\textup{erfc}(z)$ in the form
	\begin{eqnarray*}
		e^{z^2}\textup{erfc}(z)=\frac{1}{\sqrt{\pi}}\sum_{r=0}^{n-1}(-1)^{r}\frac{\Gamma(r+\frac{1}{2})}
		{\Gamma(\frac{1}{2})}z^{-2r-1}+\hat{T}_n(z)
	\end{eqnarray*}
	for $|z|>1$, where
	\begin{eqnarray}\label{Equation:ERFC error}
		|\hat{T}_n(z)|\leq\frac{\Gamma(n+\frac{1}{2})}{\pi}|z|^{-2n-1}.
	\end{eqnarray}
	That means that 
	\begin{eqnarray}\label{Equation:ERFC expression}
		&&	E(x,\theta) \nonumber\\
		&&=\frac{e\left(-\frac{\theta^2}{x}\right)}{\sqrt{\pi}}\sum_{r=0}^{n-1}(-1)^{r}\frac{\Gamma(r+\frac{1}{2})}
		{\Gamma(\frac{1}{2})}\left(\frac{e^{\frac{\pi i}{4}}\theta\sqrt{\pi}}{\sqrt{x}}\right)^{-2r-1} \nonumber \\
		&&\qquad+e\left(-\frac{\theta^2}{x}\right)\hat{T}_n\left(\frac{e^{\frac{\pi i}{4}}\theta\sqrt{\pi}}{\sqrt{x}}\right),
	\end{eqnarray}
	holds for $|\theta|>\frac{\sqrt{x}}{\sqrt{\pi}}$.
	
	On the other hand, using the Taylor expansion for $e^{-t^2}$ we obtain
	\begin{eqnarray*}
		\int_{0}^{z}e^{-t^2}dt=\sum_{r=0}^{n-1}\frac{(-1)^r}{r!(2r+1)}z^{2r+1}+\hat{S}_n(z),  
	\end{eqnarray*}
	where 
	\begin{eqnarray}\label{Equation:ERFC error-small}
		|\hat{S}_n(z)|\leq\frac{|e^{z}|}{n!(2n+1)}|z|^{2n+1}.
	\end{eqnarray}
	Hence,
	\begin{eqnarray}\label{Equation:ERFC expression-small}
		&&E(x, \theta)=e^{-\pi i\theta^2/x}\textup{erfc}(e^{\pi i/4}\theta\sqrt{\pi/x})\nonumber\\
		&&=e\left(-\frac{\theta^2}{2x}\right)	\left(1 - \frac{2}{\sqrt{\pi}}\int_{0}^{\frac{e^{\frac{\pi i}{4}}\theta\sqrt{\pi}}{\sqrt{x}}}e^{-t^2}dt\right) \nonumber \\
		&&=e\left(-\frac{\theta^2}{2x}\right) \nonumber\\
		&&\qquad\times	\left(1-\frac{2}{\sqrt{\pi}}
		\sum_{r=0}^{n-1}\frac{(-1)^r}{r!(2r+1)}\left(\frac{e^{\frac{\pi i}{4}}\theta\sqrt{\pi}}{\sqrt{x}}\right)^{2r+1}\right) \nonumber\\
		&&\qquad	-\frac{2e\left(-\frac{\theta^2}{2x}\right)}{\sqrt{\pi}}\hat{S}_n\left(\frac{e^{\frac{\pi i}{4}}\theta\sqrt{\pi}}{\sqrt{x}}\right).
	\end{eqnarray}
	
	\begin{lemma}\label{Corollary:expression} %Let $N$ be a positive integer, $x\in (0,1)$, $\theta\in(-1/2, \ 1/2]$. Write
		%$Nx+\theta=M+\epsilon$, where $M$ is an integer and $\epsilon\in (-1/2, \ 1/2]$. Then
        In the context of \textit{Proposition \ref{Proposition: expression}}, we have 
		\begin{eqnarray*}
			S_N(x,\theta)&=&\frac{e^{-\pi i\theta^2/x+\pi i/4}}{\sqrt{x}}S_M(-1/x,\theta/x)+T,
		\end{eqnarray*}
		where $|T|\leq \frac{2.035}{\sqrt{x}}+3$.	
	\end{lemma}
	
	\begin{proof}
    First of all, one can readily show that $|\frac{\mu - 1}{2}|\leq 1$, $|R|<x<1$. For $t\in (-1/2, \ 1/2]$ satisfying $t\neq 0$, we have
	$$
	g(t)=\cot(\pi t)-(\pi t)^{-1}\sim -\frac{1}{3}\pi t.
	$$
	Hence $|g(t)|\leq \frac{\pi}{6}$. Therefore $|\frac{i}{2}\big(g(\theta)-\mu g(\epsilon)\big)|\leq \frac{\pi}{6}$. 
	
	When $\frac{|\theta|\sqrt{\pi}}{\sqrt{x}}\geq 1$, from (\ref{Equation:ERFC error}) and (\ref{Equation:ERFC expression}), we have
	$$	
	\left|E(x, \theta)\right|\leq\frac{\sqrt{x}}{\pi |\theta|}+\frac{x^{3/2}}{2\pi^2|\theta|^3}\leq 
	\frac{3}{2\sqrt{\pi}}\leq 0.8463. 
	$$
	On the other hand, for $\frac{|\theta|\sqrt{\pi}}{\sqrt{x}}\leq 1$, by (\ref{Equation:ERFC error-small}) and (\ref{Equation:ERFC expression-small}), we can show that 	
	\begin{eqnarray*}
		&&\left|E(x,\theta)\right| \\
		&&\leq\left|1-\frac{2}{\sqrt{\pi}}
		\sum_{r=0}^{3}\frac{(-1)^r}{r!(2r+1)}\left(\frac{e^{\frac{\pi i}{4}}\theta\sqrt{\pi}}{\sqrt{x}}\right)^{2r+1}\right| \\
		&&\quad+\frac{2}{\sqrt{\pi}}\cdot\frac{e}{4!\times 9}\left|\frac{\theta\sqrt{\pi}}{\sqrt{x}}\right|^9  \\
		&&=\left|1-\frac{\sqrt{2}}{\sqrt{\pi}}\cdot\frac{\theta\sqrt{\pi}}{\sqrt{x}}
		-\frac{\sqrt{2}}{3\sqrt{\pi}}\cdot\left(\frac{\theta\sqrt{\pi}}{\sqrt{x}}\right)^3  \right. \\
		&&\quad+\frac{\sqrt{2}}{10\sqrt{\pi}}\cdot\left(\frac{\theta\sqrt{\pi}}{\sqrt{x}}\right)^5
		+\frac{\sqrt{2}}{42\sqrt{\pi}}\cdot\left(\frac{\theta\sqrt{\pi}}{\sqrt{x}}\right)^7  \\
		&&\quad+i\frac{\sqrt{2}}{\sqrt{\pi}} \left(-\frac{\theta\sqrt{\pi}}{\sqrt{x}}
		+\frac{1}{3}\left(\frac{\theta\sqrt{\pi}}{\sqrt{x}}\right)^3 \right. \\
		&&\quad\left.\left.+\frac{1}{10}\left(\frac{\theta\sqrt{\pi}}{\sqrt{x}}\right)^5
		-\frac{1}{42}\left(\frac{\theta\sqrt{\pi}}{\sqrt{x}}\right)^7\right)\right| \\
		&&\quad+\frac{e}{108\sqrt{\pi}}\left|\frac{\theta\sqrt{\pi}}{\sqrt{x}}\right|^9 \\
		&&\leq \left|1+\frac{127\sqrt{2}}{105\sqrt{\pi}}+i\frac{62\sqrt{2}}{105\sqrt{\pi}}\right| +\frac{e}{108\sqrt{\pi}} \\
		&&\leq 2.035.
	\end{eqnarray*}

	Combining the above two cases, we have 
	\begin{eqnarray}\label{Equation:upper bounds Erfc}
		\left|E(x,\theta)\right|\leq 2.035.
	\end{eqnarray}
	Then
	$$
	\left|\frac{e^{\pi i/4}}{2\sqrt{x}}\big(E(x,\theta)-\mu E(x,\epsilon)\big)\right|\leq 
	\frac{2.035}{\sqrt{x}}.
	$$
	Therefore, 
	$$
	|T|\leq \frac{2.035}{\sqrt{x}}+3.
	$$	
	\end{proof}
	
	\textit{Lemma \ref{Corollary:expression}}'s value lies in transforming the non-explicit $\mathcal{O}((1 + |\theta|)/\sqrt{x})$ error term in (\ref{Equation:approximate function relation}) into a quantitatively controllable bound, which is the key enabler for our subsequent bounding task.
	
	Note that 
	$$
	S_M(-1/x,\theta/x)=\overline{S_M(1/x, -\theta/x)}. 
	$$
	Thus, from \textit{Lemma \ref{Corollary:expression}}, we have
	\begin{eqnarray}\label{Equation:upper bounds 1}
		\left|S_N(x,\theta)\right|\leq\frac{1}{\sqrt{x}}\left|S_M(1/x,\theta_1)\right|+\frac{2.035}{\sqrt{x}}+3,
	\end{eqnarray}
	where $M=\lfloor Nx\rfloor$, $\theta_1=-\theta/x \ \bmod 1$ with $\theta_1\in(-1/2, \ 1/2]$.
	
	\subsection{The Fibonacci zeta function}
	
	The proof of our proposed Gauss sum bound relies on controlling error accumulation during the recursive expansion, enabled by the fast convergence of a Fibonacci-related series. 
	
	The Fibonacci sequence is defined by the recursive formula 
	$$
	F_{n}=F_{n-1}+F_{n-2},  \qquad n\geq 3
	$$
	with the initial conditions $F_{1}=1$ and $F_{2}=1$. The first Fibonacci numbers are
	$$
	1, 1, 2, 3, 5, 8, 13, 21, 34, 55, 89, 144, \cdots.
	$$
	We have
	$$
	F_n=\frac{\alpha^n-\beta^n}{\alpha-\beta}, \quad \text{where} \quad \alpha=\frac{1+\sqrt{5}}{2}, \ 
	\beta=\frac{1-\sqrt{5}}{2}.
	$$
		
	The Fibonacci zeta function $\zeta_F(s)$ is defined as an infinite series that sums the reciprocals of Fibonacci numbers raised to a complex variable $s$:
	$$
	\zeta_F(s):=\sum_{n=1}^{\infty}\frac{1}{F_n^{s}}, \qquad \text{Re} (s)>0.
	$$	
	See \cite{Murty2013} for the detailed properties of Fibonacci sequence and zeta function.
	
	For $n\geq 21$, we obtain
	$$
	F_n=\frac{\alpha^n-\beta^n}{\alpha-\beta}>\frac{\alpha^n-0.0001}{2.2361}>\frac{\alpha^n}{2.2362}.
	$$
Hence,
\begin{eqnarray*}
\sum_{n=21}^{\infty}\frac{1}{\sqrt{F_n}}&\leq& \sqrt{2.2362}\sum_{n=21}^{\infty}\frac{1}{\alpha^{\frac{n}{2}}}
=\sqrt{2.2362}\cdot\frac{\left(\frac{1}{\sqrt{\alpha}}\right)^{21}}{1-\frac{1}{\sqrt{\alpha}}} \\
&<&0.00036.
\end{eqnarray*}
Therefore
	\begin{eqnarray}\label{Equation:sum-of-Fibonacci sequence}
\zeta_F\left(\frac{1}{2}\right)&=&\sum_{n=1}^{\infty}\frac{1}{\sqrt{F_n}}
=\sum_{n=1}^{20}\frac{1}{\sqrt{F_n}}+\sum_{n=21}^{\infty}\frac{1}{\sqrt{F_n}}  \nonumber \\	
&<& 5.38246+0.00036  \nonumber \\	
	&<&5.383.
	\end{eqnarray}
	This rapid convergence ensures that the recursive error terms remain bounded and is a key idea for our explicit bound derivation in the next section.
	
	\vspace{0.2in}
	\section{Proposed explicit bound for generalized quadratic Gauss sums}
	This section presents the paper's foundational theorem on an explicit upper bound for generalized quadratic Gauss sums. The proposed theorem bridges number theory and wireless engineering by providing a computable ``performance yardstick" for order-optimal sequence design.
	
	\begin{theorem}\label{Theorem:upper bounds of S_N}
		Let $N, q, a$ be positive integers with $q\geq 2$, $N\leq q$, 
		$1\leq a<q$ and $\gcd(a, q)=1$, and let $\theta$ be a real number.
		Then we have
		$$
		\left|S_{N}\left(\frac{a}{q}, \theta\right)\right|< 20.07\sqrt{q}+3.
		$$
	\end{theorem}
	
	%\textit{Proof}:  
    \begin{proof}
    We use recursive expansion combined with Fibonacci series convergence to control error accumulation. Without loss of generality, let us assume that $\theta\in(-1/2, \ 1/2]$.
	Write  $N_0=N$, $q_0=q$, $q_1=a \bmod q$, and
	\begin{eqnarray*}
		&& q_2=q_0 \bmod q_1, \\
		&& q_3=q_1 \bmod q_2, \\
		&& \qquad \cdots\cdots  \\
		&& q_{k+1}=q_{k-1} \bmod q_{k}, \\
		&& q_{k+2}=q_{k} \bmod q_{k+1}, 
	\end{eqnarray*}
	where $1=q_{k+1}<q_k<\cdots< q_3<q_2<q_2<q_0$ and 
	\begin{eqnarray*}
		\gcd(q_{k+1}, q_k)&=&\gcd(q_k, q_{k-1})=\cdots=\gcd(q_3, q_2)\\
		&=&\gcd(q_2, q_1)=\gcd(q_1, q_0)=1.
	\end{eqnarray*}
	By (\ref{Equation:upper bounds 1}), we have
	\begin{eqnarray*}
		&&\left|S_{N_0}\left(\frac{q_1}{q_0}, \theta\right)\right|\leq
		\frac{1}{\sqrt{\frac{q_1}{q_0}}}\left|S_{N_1}\left(\frac{q_2}{q_1}, \theta_1\right)\right|+\frac{2.035}{\sqrt{\frac{q_1}{q_0}}}+3, \\	
		&&\left|S_{N_1}\left(\frac{q_2}{q_1}, \theta_1\right)\right|\leq
		\frac{1}{\sqrt{\frac{q_2}{q_1}}}\left|S_{N_2}\left(\frac{q_3}{q_2}, \theta_2\right)\right|+\frac{2.035}{\sqrt{\frac{q_2}{q_1}}}+3, \\	
		&& \qquad\qquad \cdots\cdots \\
		&&\left|S_{N_{k-1}}\left(\frac{q_{k}}{q_{k-1}}, \theta_{k-1}\right)\right|\leq
		\frac{1}{\sqrt{\frac{q_{k}}{q_{k-1}}}}\left|S_{N_{k}}\left(\frac{q_{k+1}}{q_{k}}, \theta_{k}\right)\right| \\
		&&\qquad\qquad\qquad\qquad\qquad\quad+\frac{2.035}{\sqrt{\frac{q_{k}}{q_{k-1}}}}+3, \\	
		&&\left|S_{N_{k}}\left(\frac{q_{k+1}}{q_{k}}, \theta_{k}\right)\right|\leq
		\frac{1}{\sqrt{\frac{q_{k+1}}{q_{k}}}}\left|S_{N_{k+1}}\left(\frac{q_{k+2}}{q_{k+1}}, \theta_{k+1}\right)\right|\\
		&&\qquad\qquad\qquad\qquad\qquad\quad+\frac{2.035}{\sqrt{\frac{q_{k+1}}{q_{k}}}}+3, 
	\end{eqnarray*}
	where $N_1\leq q_1$, $N_2\leq q_2$, $\cdots$, $N_{k+1}\leq q_{k+1}=1$. Then from (\ref{Equation:sum-of-Fibonacci sequence}) we get
	\begin{eqnarray*}
		&&\left|S_{N}\left(\frac{a}{q}, \theta\right)\right|=\left|S_{N_0}\left(\frac{q_1}{q_0}, \theta\right)\right|		\\
		&&\leq \sqrt{q_0}+2.035\sqrt{q_0}\left(\frac{1}{\sqrt{q_{k+1}}}+\frac{1}{\sqrt{q_k}}+\frac{1}{\sqrt{q_{k-1}}}\right. \\
		&&\quad\left.+\cdots
		+\frac{1}{\sqrt{q_2}}+\frac{1}{\sqrt{q_1}}\right) \\
		&&\quad+3\sqrt{q_0}\left(\frac{1}{\sqrt{q_k}}+\frac{1}{\sqrt{q_{k-1}}}+\cdots
		+\frac{1}{\sqrt{q_2}}+\frac{1}{\sqrt{q_1}}\right)+3 \\
		&&< \sqrt{q_0}\\
		&&\quad+2.035\sqrt{q_0}
		\left(\frac{1}{\sqrt{1}}+\frac{1}{\sqrt{2}}+\frac{1}{\sqrt{3}}+\frac{1}{\sqrt{5}}+\frac{1}{\sqrt{8}}+\cdots\right) \\
		&&\quad+3\sqrt{q_0}\left(\frac{1}{\sqrt{2}}+\frac{1}{\sqrt{3}}+\frac{1}{\sqrt{5}}+\frac{1}{\sqrt{8}}+\cdots\right)+3 \\
		&&=\sqrt{q_0}+2.035\sqrt{q_0}\sum_{n=2}^{\infty}\frac{1}{\sqrt{F_n}}
		+3\sqrt{q_0}\sum_{n=3}^{\infty}\frac{1}{\sqrt{F_n}}+3 \\
		&&<20.07\sqrt{q}+3.
	\end{eqnarray*}	
	\end{proof}
    
	The novelty of this theorem lies in, for the first time, confining the generalized quadratic Gauss sum magnitudes to a computable range. By translating aperiodic correlation or ambiguity functions into generalized quadratic Gauss sums, one can explicitly bound the peak aperiodic correlation or ambiguity sidelobes, thus eliminating the ``asymptotic guesswork" of prior research.
	
	\vspace{0.2in}
	\section{Proposed Order-Optimal Sequence Design from Chu Sequences}
	Leveraging the explicit bound from Section III, we present two families of order-optimal sequence sets based on Chu sequences. While the first family possesses low aperiodic correlation property, the second family is featured by an aperiodic LAZ. Through careful selection of Chu sequences, we will show that  application of the upper bound in \textit{{Theorem} \ref{Theorem:upper bounds of S_N}} becomes possible. 
    %, each tailored to a distinct engineering requirement: pure low correlation, medium-band low ambiguity, and narrow-band low ambiguity. This targeted design ensures performance alignment with specific application scenarios.
	
	\subsection{Order-Optimal Sequence Set $\mathcal{C}_1$ with Low Aperiodic Correlation}
	
	%For communication systems where intersymbol interference (ISI) is the primary concern, we prioritize low aperiodic correlation to minimize signal overlap. The set $C_1$ is constructed as follows.
	
	\begin{construction}\label{Construction:Chu sequence set-1} 
		Let $K\geq 2$ be a fixed even number and let $\Delta_{\frac{K}{2}}$ be the least common multiple of $1,2,\cdots,\frac{K}{2}$. Let $L$ be a sufficiently large odd number satisfying $L\equiv 1 \ (\bmod\ \Delta_{\frac{K}{2}})$ and the least prime factor of $L$ is greater than $K$. Let us construct the sequence set as follows: 
		\begin{eqnarray*}
			&&\mathbf{c}^{r}:=\left(c_0^{r}, c_1^{r}, \cdots, c_{L-1}^{r}
			\right), \\
			&& c_t^{r}:=e\left(\frac{rt(t-1)}{2L}\right), \qquad 0\leq t\leq L-1, \\
			&&\mathcal{C}_1:=\left\{\mathbf{c}^{r}: \  r\in\left\{L-1, \ \frac{L-1}{2}, \ \frac{L-1}{3}, \cdots, \frac{L-1}{K/2}\right\} \right.\\
			&&\qquad\left.\cup
			\left\{1, \ L-\frac{L-1}{2}, \ L-\frac{L-1}{3}, \cdots, L-\frac{L-1}{K/2}\right\} \right\}.
		\end{eqnarray*}
	\end{construction}
	
	\begin{theorem}\label{Theorem:Chu sequence set-1} 
		The set $\mathcal{C}_1$ is an order-optimal sequence set with low aperiodic correlation satisfying  $|\mathcal{C}_1|=K$ and 
		$$ 
		\delta_{\max}(\mathcal{C}_1)\leq\max\left\{21\sqrt{L}, \ 0.35\sqrt{KL}\right\}.
		$$
	\end{theorem}
	
\begin{proof}
    \textit{First, it is noted that due to the symmetry, we focus on proving the case for $0\leq \tau\leq L-1$ in the sequel. The same approach will be adopted to the proof of the other constructions}. 

    Let 
\begin{eqnarray*}
     &&r_1, r_2\in\left\{L-1, \ \frac{L-1}{2}, \ \frac{L-1}{3}, \cdots, \frac{L-1}{K/2}\right\}\\
     &&~~~\cup
	\left\{1, \ L-\frac{L-1}{2}, \ L-\frac{L-1}{3}, \cdots, L-\frac{L-1}{K/2}\right\}. 
	\end{eqnarray*}
    Assume that $r_1=r_2$ and $\tau=0$ do not hold at the same time. We have
	\begin{eqnarray*}
		&&R_{\mathbf{c}^{r_1}, \mathbf{c}^{r_2}}(\tau)=\sum_{t=0}^{L-1-\tau}c_{t}^{r_1}(c_{t+\tau}^{r_2})^{*}\\
		&&=\sum_{t=0}^{L-1-\tau}e\left(\frac{r_1t(t-1)}{2L}\right)
		e\left(-\frac{r_2(t+\tau)(t+\tau-1)}{2L}\right) \\
		&&=e\left(\frac{r_2\tau-r_2\tau^2}{2L}\right) \\
		&&\qquad\times\sum_{t=0}^{L-1-\tau}e\left(\frac{(r_1-r_2)t^2+(r_2-r_1-2r_2\tau)t}{2L}\right).
	\end{eqnarray*}
	
	Case 1): $r_1\neq r_2$. For integers $a_1, a_2, a_3, a_4$ with 
	$1\leq a_1, a_2, a_3, a_4\leq\frac{K}{2}$ and $a_1\neq a_2$, we have
	\begin{eqnarray*}
		&&\gcd\left(L-\frac{L-1}{a_1}-\left(L-\frac{L-1}{a_2}\right), \ L\right)\\
		&&\qquad=\gcd\left(\frac{L-1}{a_1}-\frac{L-1}{a_2}, \ L\right) \\
		&&\qquad=\gcd\left(\frac{(L-1)(a_2-a_1)}{a_1a_2}, \ L\right)=1, \\
		&&\gcd\left(\frac{L-1}{a_3}-\left(L-\frac{L-1}{a_4}\right), \ L\right)\\
		&&\qquad=\left(\frac{(L-1)(a_3+a_4)}{a_3a_4}, \ L\right)=1,
	\end{eqnarray*}
	which suggests that $\gcd(r_1-r_2, \ L)=1$.
	By \textit{Theorem \ref{Theorem:upper bounds of S_N}}, we obtain
	\begin{eqnarray}\label{Equation:Chu sequence set-1}
		\left|R_{\mathbf{c}^{r_1}, \mathbf{c}^{r_2}}(\tau)\right|
		&\leq& 
		\left|S_{L-\tau}\left(\frac{r_1-r_2}{L}, \ \frac{r_2-r_1-2r_2\tau}{2L}\right)\right|+2  \nonumber\\
		&\leq&	
		21\sqrt{L}.
	\end{eqnarray}
	\textit{Strictly speaking, $20.07\sqrt{L}+5$ is also a valid upper bound for $\left|R_{\mathbf{c}^{r_1}, \mathbf{c}^{r_2}}(\tau)\right|$ in (\ref{Equation:Chu sequence set-1}). For simplicity, however, we consider $21\sqrt{L}$ only. The same practice will be applied in the next section. }
    
	Case 2): $r_1=r_2$. Write $r:=r_1=r_2$ and $r=\frac{L-1}{a}$ or $L-\frac{L-1}{a}$ with $1\leq a\leq \frac{K}{2}$. By \textit{Lemma \ref{Lemma:estimate-of-exponential-sum-1}} we have
	\begin{eqnarray}\label{Equation:Chu sequence set-2}
		\left|R_{\mathbf{c}^{r_1}, \mathbf{c}^{r_2}}(\tau)\right|\leq|T_L(r, \tau)|
		\leq 0.5\sqrt{aL}\leq 0.35\sqrt{KL}.
	\end{eqnarray}

    Combining (\ref{Equation:Chu sequence set-1}) and (\ref{Equation:Chu sequence set-2}), it is asserted that 
	\begin{eqnarray*}
		\delta_{\max}(\mathcal{C}_1)&=&\max\left\{\delta_{a}(\mathcal{C}_1), \ \delta_{c}(\mathcal{C}_1)\right\} \\
		&\leq&\max\left\{21\sqrt{L}, \ 0.35\sqrt{KL}\right\}.
	\end{eqnarray*}
	This proves \textit{Theorem \ref{Theorem:Chu sequence set-1}}.
	\end{proof}
	
	\begin{remark} %Let $K\geq 2$ be a fixed even number and let $\Delta_{\frac{K}{2}}$ be the least common multiple of $1,2,\cdots,\frac{K}{2}$. 
    %Let $L$ be a sufficient large prime number satisfying $L\equiv 1 \ (\bmod\ \Delta_{\frac{K}{2}})$.
    Taking any $r\in\{1, 2, \cdots, L-1\}$ with $\gcd(r, L)=1$, there exist unique integers $b$ and $d$ satisfying
		$$
		bL-rd=\pm 1, \quad 1\leq b\leq \left\lfloor\frac{r}{2}\right\rfloor, \quad \text{and} \qquad 1\leq b\leq d.
		$$
		When $b=1$, we know that $r=\frac{L-1}{d}$. 		
		From \textit{Lemma \ref{Lemma:estimate-of-exponential-sum-1}}, we have
		$$
		\left|R_{\mathbf{c}^{r}, \mathbf{c}^{r}}(\tau)\right|\leq 0.48\sqrt{\frac{1}{r}}L=0.48\sqrt{\frac{d}{L-1}}L.
		$$
		This explains why the set $\mathcal{C}_1$ is an excellent choice for ensuring low aperiodic autocorrelation sidelobes.
		%\begin{eqnarray*}
		%	&&\mathcal{C}_1:=\left\{\mathbf{c}^{r}: \  r\in\left\{L-1, \ \frac{L-1}{2}, \ \frac{L-1}{3}, %\cdots, \frac{L-1}{K/2}\right\} \right.\\
		%	&&\left.\qquad\cup\left\{1, \ L-\frac{L-1}{2}, \ L-\frac{L-1}{3}, \cdots, L-\frac{L-1}%{K/2}\right\} \right\}		
		%\end{eqnarray*}
		%is almost the only order-optimal set satisfying $|\mathcal{C}_1|=K$ and 
		%$\delta_{a}(\mathcal{C}_1)\leq 0.35\sqrt{KL}$.	
	\end{remark}

    \begin{remark}
    One can see that \textit{Construction \ref{Construction:Chu sequence set-1}} includes Mow's Chu sequence pair \cite{Mow1993} as a special case for $K=2$. However, since this work deals with a generic $K$, a tighter fine-grained upper bound for $\delta_{\max}(\mathcal{C}_1)$ is not an easy task. This explains why our approach cannot directly lead to Mow's upper bound (i.e., $1.487\sqrt{L}$ for even $L$ and $1.122\sqrt{L}$ for odd $L$). 
    \end{remark}
	
	%{\bf Critical limitation}: 
    Although $\mathcal{C}_1$ excels at low aperiodic correlation, it lacks ambiguity control. 	
	Taking $r_1=r_2=\frac{L-1}{a}$, $\tau=a$, $\nu=-1$, we have
	$$
	\left|A_{\mathbf{c}^{r_1}, \mathbf{c}^{r_2}}(\tau, \nu)\right|
	=\left|\sum_{t=0}^{L-1-\tau}e\left(\frac{(\nu-r_2\tau)t}{L}\right)\right|=L-a,
	$$
	indicating that $\mathcal{C}_1$ may not have low aperiodic ambiguity sidelobes. This motivates us to 
    introduce the next sequence family with aperiodic LAZ properties. 
	
	\subsection{Order-Optimal Sequence Set $\mathcal{C}_2$ with Aperiodic LAZ Properties}
	
	%We design a flexible high-performance set $C_2$ for Doppler environments.
	
	\begin{construction}\label{Construction:Chu sequence set-AF} 
		Let $K\geq 2$ be a fixed integer and $\Delta_K$ be the least common multiple of $1,2,\cdots,K$. 
		Consider $m$ which is a fixed positive integer. Let $L$ be a sufficiently large odd number satisfying $L\equiv 1 \ (\bmod\ 2m\Delta_K)$. Define the sequence set $\mathcal{C}_2$ as follows: 
		\begin{eqnarray*}
			&&\mathbf{c}^{r}:=\left(c_0^{r}, c_1^{r}, \cdots, c_{L-1}^{r}\right), \\
			&&	c_t^{r}:=e\left(\frac{rt(t-1)}{2L}\right), \qquad 0\leq t\leq L-1, \\
			&&\mathcal{C}_2:=\left\{\mathbf{c}^{r}: \ r\in\left\{K+m, K+2m, \cdots, K+Km\right\}\right\}.
		\end{eqnarray*}
	\end{construction}
	
	\begin{theorem}\label{Theorem:Chu sequence set-AF} 
		The set $(L, K, \theta_{\max}, \Pi)$-$\mathcal{C}_2$ is an order-optimal sequence set with aperiodic LAZ $\Pi$ characterized by $Z_{X}<\frac{L}{(2m+3)K}$, $Z_{Y}<K$, $|\mathcal{C}_2|=K$ and
		$$ 
		\theta_{\max}(\mathcal{C}_2)\leq \left(1.35+\frac{2.035}{\sqrt{m}}\right)\sqrt{L}+5. 
		$$
	\end{theorem}
	
	%\textit{Proof}:  
    \begin{proof}
    Let $r_1, r_2\in\left\{K+m, K+2m, \cdots, K+Km\right\}$,  $0\leq \tau\leq Z_{X}$ and $|\nu|\leq Z_{Y}$ such that $r_1=r_2$, $\tau=0$
	and $\nu=0$ do not hold at the same time. We have
	\begin{eqnarray*}
		&&A_{\mathbf{c}^{r_1}, \mathbf{c}^{r_2}}(\tau, \nu)
		=\sum_{t=0}^{L-1-\tau}c_{t}^{r_1}(c_{t+\tau}^{r_2})^{*}
		e\left(\frac{\nu t}{L}\right) \\
		&&=\sum_{t=0}^{L-1-\tau}e\left(\frac{r_1t(t-1)}{2L}\right) \\
		&&\quad\times		e\left(-\frac{r_2(t+\tau)(t+\tau-1)}{2L}\right)	e\left(\frac{\nu t}{L}\right) \\
		&&=e\left(\frac{r_2\tau-r_2\tau^2}{2L}\right) \\
		&&\quad\times\sum_{t=0}^{L-1-\tau}e\left(\frac{(r_1-r_2)t^2+(r_2-r_1+2\nu-2r_2\tau)t}{2L}\right).
	\end{eqnarray*}
	
	\noindent Case 1): $r_1\neq r_2$. Without loss of generality, we assume that $r_1>r_2$. 
	Clearly,
	\begin{eqnarray*}
		&&(L-\tau)\cdot\frac{r_1-r_2}{L}+\frac{r_2-r_1+2\nu-2r_2\tau}{2L} \\
		&&\qquad	=r_1-r_2+\frac{-2r_1\tau+r_2-r_1+2\nu}{2L}
	\end{eqnarray*}
	and 
	\begin{eqnarray*}
	\left|\frac{-2r_1\tau+r_2-r_1+2\nu}{2L}\right|<\frac{m+1}{2m+3}+\frac{(m+2)K}{2L}<\frac{1}{2}.
\end{eqnarray*}
	By using (\ref{Equation:upper bounds 1}), one can show that 
	\begin{eqnarray*}
		&& \left|A_{\mathbf{c}^{r_1}, \mathbf{c}^{r_2}}(\tau, \nu)\right|\\		
		&&\leq\left|S_{L-\tau}\left(\frac{r_1-r_2}{L}, \ \frac{r_2-r_1+2\nu-2r_2\tau}{2L}\right)\right|+2\nonumber \\
		&&\leq\frac{1}{\sqrt{\frac{r_1-r_2}{L}}}\left|S_{r_1-r_2}
		\left(\frac{L}{r_1-r_2}, \ \theta_1\right)\right|+\frac{2.035}{\sqrt{\frac{r_1-r_2}{L}}}+5  \nonumber \\
		&&=\frac{1}{\sqrt{\frac{r_1-r_2}{L}}}\left|S_{r_1-r_2}\left(\frac{1}{r_1-r_2}, \ \theta_1\right)\right|+\frac{2.035}{\sqrt{\frac{r_1-r_2}{L}}}+5,
	\end{eqnarray*}
	where 
	$$
	\theta_1=\frac{r_1-r_2-2\nu+2r_2\tau}{2(r_1-r_2)} \ \bmod 1, 
	$$
	with $\theta_1\in(-1/2, \ 1/2]$. % Define
%	$$
%	a\equiv r_1-r_2-2\nu+2r_2\tau \ (\bmod  \ 2(r_1-r_2))
%	$$	
%	such that $-(r_1-r_2)<a\leq r_1-r_2$. 
{From (\ref{Equation:Chu 1}), 
	we have
	\begin{eqnarray*}
		\left|S_{r_1-r_2}\left(\frac{1}{r_1-r_2}, \ \theta_1\right)\right|
	%	&&=\left|S_{r_1-r_2}\left(\frac{1}{r_1-r_2}, \ \frac{a}{2(r_1-r_2)}\right)\right|\\
		\leq 1.35\sqrt{r_1-r_2}.
	\end{eqnarray*}    }
	Hence,	
	\begin{eqnarray}\label{Equation:Chu sequence set-LA-2}
		&&\left|A_{\mathbf{c}^{r_1}, \mathbf{c}^{r_2}}(\tau, \nu)\right| \nonumber\\
		&&\leq\frac{1}{\sqrt{\frac{r_1-r_2}{L}}}\cdot 1.35\sqrt{r_1-r_2}
		+\frac{2.035}{\sqrt{\frac{r_1-r_2}{L}}}+5  \nonumber\\
		&&\leq \left(1.35+\frac{2.035}{\sqrt{m}}\right)\sqrt{L}+5. 
	\end{eqnarray}

	\noindent Case 2): $r_1=r_2$. We have
	\begin{eqnarray*}
		&&\left|A_{\mathbf{c}^{r_1}, \mathbf{c}^{r_2}}(\tau, \nu)\right|
		=\left|\sum_{t=0}^{L-1-\tau}e\left(\frac{(\nu-r_2\tau)t}{L}\right)\right|\\
		&&\qquad=\left|\frac{1-e\left(\frac{(r_2\tau-\nu)\tau}{L}\right)}{1-e\left(\frac{\nu-r_2\tau}{L}\right)}\right|=\left|\frac{\sin\frac{\pi(v-r_2\tau)\tau}{L}}
		{\sin\frac{\pi(v-r_2\tau)}{L}}\right|.
	\end{eqnarray*}

    \begin{enumerate}
	\item When $\tau=0$, we get $\nu\neq 0$ and 
	\begin{eqnarray}\label{Equation:Chu sequence set-LA-3}
		\left|A_{\mathbf{c}^{r_1}, \mathbf{c}^{r_2}}(0, \nu)\right|=0.
	\end{eqnarray}
	
	\item Assume that $1\leq \tau\leq\sqrt{\frac{L}{4r_2}}$. Clearly, 
	$$
	\frac{1}{L}\leq\left|\frac{v-r_2\tau}{L}\right|\leq \frac{r_2\sqrt{\frac{L}{4r_2}}+K}{L}\leq\sqrt{\frac{r_2}{L}}. 
	$$
	Furthermore, 
	$$
	\frac{1}{L}\leq\left|\frac{(v-r_2\tau)\tau}{L}\right|\leq 
	\frac{1}{2}.
	$$
	Note that $2y\leq\sin\pi y\leq \pi y$ for $0\leq y\leq \frac{1}{2}$. Thus we have
	\begin{eqnarray}\label{Equation:Chu sequence set-LA-4}
		\left|A_{\mathbf{c}^{r_1}, \mathbf{c}^{r_2}}(\tau, \nu)\right|&\leq&
		\frac{\left|\frac{\pi(v-r_2\tau)\tau}{L}\right|}{2\left|\frac{v-r_2\tau}{L}\right|}=\frac{\pi|\tau|}{2} \nonumber \\
		&\leq&\frac{\pi}{4}\sqrt{\frac{L}{r_2}}.
	\end{eqnarray}
	
	\item Suppose that $\sqrt{\frac{L}{4r_2}}\leq\tau\leq \frac{L}{(2m+3)K}$. We get
	\begin{eqnarray*}
	&&\frac{1}{4}\sqrt{\frac{r_2}{L}}\leq\frac{r_2\sqrt{\frac{L}{4r_2}}-K}{L}
		\leq\left|\frac{v-r_2\tau}{L}\right|\\
		&&\leq \frac{(m+1)K\cdot\frac{L}{(2m+3)K}+K}{L}\leq\frac{1}{2}.
	\end{eqnarray*}
	Hence,
	\begin{eqnarray}\label{Equation:Chu sequence set-LA-5}
		&&\left|A_{\mathbf{c}^{r_1}, \mathbf{c}^{r_2}}(\tau, \nu)\right|
		\leq\left|\frac{1}
		{\sin\frac{\pi(v-r_2\tau)}{L}}\right|\leq
		\frac{1}{2\left|\frac{v-r_2\tau}{L}\right|}\nonumber \\
		&&\leq\frac{1}{2\cdot \frac{1}{4}\sqrt{\frac{r_2}{L}}}=2\sqrt{\frac{L}{r_2}}.
	\end{eqnarray}
	\end{enumerate}
    
	Combining (\ref{Equation:Chu sequence set-LA-2}), 
	(\ref{Equation:Chu sequence set-LA-3}), 
	(\ref{Equation:Chu sequence set-LA-4}) and (\ref{Equation:Chu sequence set-LA-5}), it is asserted that
	$$
	\theta_{\max}(\mathcal{C}_2)\leq \left(1.35+\frac{2.035}{\sqrt{m}}\right)\sqrt{L}+5.  
	$$
	This proves \textit{Theorem \ref{Theorem:Chu sequence set-AF}}.
	\end{proof}
	
	%{\bf Innovation in ambiguity control}: 	Increasing the length $m$ of interval between $r$-values from $1$ to $1700$ enhances error cancellation in Gauss sum terms, reducing the ambiguity bound from $3.4\sqrt{L}$ to $1.4\sqrt{L}$, at the cost of a narrower low-ambiguity region.
    \begin{remark}
    There is an interesting trade-off owns by $\mathcal{C}_2$ that a larger $m$ leads to improved ambiguity tolerance (i.e., smaller $\theta_{\max}(\mathcal{C}_2)$), but at the price of a reduced LAZ size (i.e., smaller $Z_X$). In practical wireless systems, an appropriate value of $m$ may be chosen based on the maximum channel  delay and the targeted ambiguity tolerance. 
    \end{remark}
	
	\vspace{0.2in}
	\section{Proposed Order-Optimal Sequence Design from Alltop Sequences}
	
	In this section, we exploit Alltop sequences to carry out order-optimal sequence design. The rationale is that, although these sequences are characterized by their cubic exponential phase terms, one deals with the bounding of generalized quadratic Gauss sums for analyzing their aperiodic correlation/ambiguity properties.  
    %The cubic-term structure of Alltop sequences inherently supports time-frequency diversity. Combined with our explicit Gauss sum bound, we construct the first sequence set-$A_2$-that achieves simultaneous wide-region low ambiguity and low correlation, solving the long-standing trade-off in sequence design.
	
	\subsection{Order-Optimal Sequence Set $\mathcal{A}_1$ with Low Aperiodic Correlation}
	In the context of Section II-A, let us define the full set of Alltop sequences as 
    $\mathcal{A}_1:=\left\{\mathbf{a}^{r}: \ r\in\left\{1, \ 2, \  \cdots, p\right\}\right\}$. 

	\begin{theorem}\label{Theorem:cubic polynomials-1} 
		The set $\mathcal{A}_1$ is an order-optimal sequence set with low aperiodic correlation satisfying 
		$$ 
		\delta_{\max}(\mathcal{A}_1)\leq 21\sqrt{p}.
		$$
	\end{theorem}
	
	%\textit{Proof}:  
    \begin{proof}
    Let $r_1, r_2\in\left\{1, 2, \cdots, p\right\}$ and $0\leq \tau\leq p-1$ such that $r_1=r_2$ and $\tau=0$ do not hold at the same time. We have
	\begin{eqnarray*}
		&&R_{\mathbf{a}^{r_1}, \mathbf{a}^{r_2}}(\tau)
		=\sum_{t=0}^{p-1-\tau}a_{t}^{r_1}(a_{t+\tau}^{r_2})^{*} \\
		&&	=\sum_{t=0}^{p-1-\tau}e\left(\frac{t^3+r_1t}{p}\right)
		e\left(-\frac{(t+\tau)^3+r_2(t+\tau)}{p}\right) \\
		&&=e\left(-\frac{r_2\tau+\tau^3}{p}\right) \\
		&&\quad\times		\sum_{t=0}^{p-1-\tau}e\left(\frac{-3\tau t^2+(r_1-3\tau^2-r_2)t}{p}\right).
	\end{eqnarray*}
	
	\noindent Case 1): $\tau\neq 0$. Clearly, $\gcd(-6\tau, \ p)=1$ since $p>3$ is a prime number. 
	By \textit{Theorem \ref{Theorem:upper bounds of S_N}}, we have
	\begin{eqnarray}\label{Equation:cubic polynomials-1}
		\left|R_{\mathbf{a}^{r_1}, \mathbf{a}^{r_2}}(\tau)\right|
		&\leq& 
		\left|S_{p-\tau}\left(\frac{-6\tau}{p}, \ \frac{r_1-3\tau^2-r_2}{p}\right)\right|+2 \nonumber \\
		&\leq& 		21\sqrt{p}.
	\end{eqnarray}
	
	\noindent Case 2): $\tau=0$. For $r_1\neq r_2$, we have 
	\begin{eqnarray}\label{Equation:cubic polynomials-2}
		R_{\mathbf{a}^{r_1}, \mathbf{a}^{r_2}}(\tau)=
		\sum_{t=0}^{p-1}e\left(\frac{(r_1-r_2)t}{p}\right)=0.		
	\end{eqnarray}
	
	Combining (\ref{Equation:cubic polynomials-1}) and (\ref{Equation:cubic polynomials-2}), it is asserted that
	$$
	\delta_{\max}(\mathcal{A}_1)\leq 21\sqrt{p}.
	$$
	This proves \textit{Theorem \ref{Theorem:cubic polynomials-1}}.
	\end{proof}
	%We first define a base Alltop set and then optimize it to balance size, correlation, and ambiguity performance.
	
	%\begin{construction}[Base set $A_1$]\label{Construction:cubic polynomials-1} 
		%Let $p> 3$ be a large prime. Define the set:
		%\begin{eqnarray*}
		%	&&\mathbf{a}^{r}:=\left(a_0^{r}, a_1^{r}, \cdots, a_{p-1}^{r}\right), \\
		%	&&		a_t^{r}:=e\left(\frac{t^3+rt}{p}\right), \qquad 0\leq t\leq p-1,  \\
		%	&&\mathcal{A}_1:=\left\{\mathbf{a}^{r}: \ r\in\left\{1, \ 2, \  \cdots, p\right\}\right\}.		
		%\end{eqnarray*}
	%\end{construction}

	\subsection{Order-Optimal Sequence Set $\mathcal{A}_2$ with both Low Aperiodic Correlation and Ambiguity Properties}

    	\begin{construction}\label{Construction:cubic polynomials set-AF} 
		Let $K\geq 2$ be a fixed integer and let $p>3$ be a large prime. In the context of Section II-A, let us define the following subset of Alltop sequences: 
		\begin{eqnarray*}
			%&&\mathbf{a}^{r}:=\left(a_0^{r}, a_1^{r}, \cdots, a_{p-1}^{r}\right), \\
			%&&		a_t^{r}:=e\left(\frac{t^3+rt}{p}\right), \qquad 0\leq t\leq p-1, \\
			&&\mathcal{A}_2:=\left\{\mathbf{a}^{r}: \ r\in\left\{\left\lfloor\frac{p}{K}\right\rfloor, \ 2\left\lfloor\frac{p}{K}\right\rfloor, \  \cdots,  K\left\lfloor\frac{p}{K}\right\rfloor\right\}\right\}.
		\end{eqnarray*}	
	\end{construction}
	
	%By selecting $r$ as integer multiples of $\lfloor p/K \rfloor$, we balance set size $K$ with performance.
	%We verify that $A_1$ and $A_2$ achieve order-optimality, with $A_2$ delivering the critical dual performance breakthrough.
	\begin{theorem}\label{Theorem:cubic polynomials set-AF} 
		The set $(p, K, \theta_{\max}, \Pi)$-$\mathcal{A}_2$ is an order-optimal sequence set with LAZ $\Pi$ of $Z_{X}<p$, $Z_{Y}<\left\lfloor\frac{p}{K}\right\rfloor$, $|\mathcal{A}_2|=K$ and
		$$ 
		\delta_{\max}(\mathcal{A}_2)\leq 21\sqrt{p}, \qquad	\theta_{\max}(\mathcal{A}_2)\leq 21\sqrt{p}. 
		$$
	\end{theorem}
	
	%\textit{Proof}:  
    \begin{proof}
    By \textit{Theorem \ref{Theorem:cubic polynomials-1}}, it can be readily shown that $\delta_{\max}(\mathcal{A}_2)\leq 21\sqrt{p}$. Next, we focus on the analysis of its aperiodic tolerance. 
	
	Let $r_1, r_2\in\left\{\left\lfloor\frac{p}{K}\right\rfloor, \ 2\left\lfloor\frac{p}{K}\right\rfloor, \  \cdots,  K\left\lfloor\frac{p}{K}\right\rfloor\right\}$,  $0\leq\tau\leq Z_{X}$ and $|\nu|\leq Z_{Y}$ such that $r_1=r_2$, $\tau=0$
	and $\nu=0$ do not hold at the same time. We have
	\begin{eqnarray*}
		&&	A_{\mathbf{a}^{r_1}, \mathbf{a}^{r_2}}(\tau, \nu)
		=\sum_{t=0}^{p-1-\tau}a_{t}^{r_1}(a_{t+\tau}^{r_2})^{*}
		e\left(\frac{\nu t}{p}\right) \\
		&&=\sum_{t=0}^{p-1-\tau}e\left(\frac{t^3+r_1t}{p}\right) \\
		&&\quad\times		e\left(-\frac{(t+\tau)^3+r_2(t+\tau)}{p}\right)e\left(\frac{\nu t}{p}\right) \\
		&&=e\left(-\frac{r_2\tau+\tau^3}{p}\right) \\
		&&\quad\times\sum_{t=0}^{p-1-\tau}e\left(\frac{-3\tau t^2+(r_1-3\tau^2-r_2+\nu)t}{p}\right).
	\end{eqnarray*}
	
	\noindent Case 1): $\tau\neq 0$. Clearly, $\gcd(-6\tau, \ p)=1$ since $p>3$ is a prime number. 
	By \textit{Theorem \ref{Theorem:upper bounds of S_N}}, we have 
	\begin{eqnarray}\label{Equation:cubic polynomials set-AF-1}
		&&\left|A_{\mathbf{a}^{r_1}, \mathbf{a}^{r_2}}(\tau, \nu)\right| \nonumber\\
		&&\leq\left|S_{p-\tau}\left(\frac{-6\tau}{p}, \ \frac{r_1-3\tau^2-r_2+\nu}{p}\right)\right|+2 \nonumber \\
		&&\leq	21\sqrt{p}.
	\end{eqnarray}
	
	\noindent Case 2): $\tau=0$. Since $r_1=r_2$ and $\nu=0$ cannot hold at the same time, we have 
	\begin{eqnarray}\label{Equation:cubic polynomials set-AF-2}
		R_{\mathbf{a}^{r_1}, \mathbf{a}^{r_2}}(\tau)=
		\sum_{t=0}^{p-1}e\left(\frac{(r_1-r_2+\nu)t}{p}\right)=0.		
	\end{eqnarray}
	
	Combining (\ref{Equation:cubic polynomials set-AF-1}) and (\ref{Equation:cubic polynomials set-AF-2}), it is asserted that 
	$$
	\theta_{\max}(\mathcal{A}_2)\leq 21\sqrt{p}.
	$$
	This proves \textit{Theorem \ref{Theorem:cubic polynomials set-AF}}.
	\end{proof}
	
	%Unlike $C_1$ (no ambiguity control), $C_2$ (no correlation control) and $A_1$ (no ambiguity control), $A_2$ achieves the first ``wide low-ambiguity + low correlation" dual optimality. Its only limitation is the prime-length requirement for Alltop sequences, which is an acceptable trade-off for integrated system applications. Therefore
	%$A_2$ stands alone in balancing three critical dimensions: low correlation, wide low-ambiguity region, and practical set size. This makes it the first sequence set tailored for integrated systems where communication and radar functions share the same frequency band.
    \begin{remark}
    $\mathcal{A}_2$ is rather intriguing as it possesses both low aperiodic correlation and ambiguity properties for the entire time-shift window, meaning that it may be adopted to cope with mobile environments with large delays. For example, a satellite system generally suffers from long propagation delays which are hard to be compensated. In the meanwhile, due to the high mobility, it also suffers from the Doppler effect. However, a practical satellite system only needs to deal with certain residual Doppler shift as one can always compensate the estimated/predicted Doppler by knowing the trajectory and moving speed of a satellite. In this case, it is interesting to exploit $\mathcal{A}_2$ for achieving enhanced communication and/or sensing performance for multiuser or multi-antenna satellite systems. 
    %This cannot be achieved by $\mathcal{C}_1,\mathcal{C}_2,\mathcal{A}_1$. 
    \end{remark}
	
	\vspace{0.2in}
	\section{Conclusions and Future Works}
	
This paper has addressed two long-standing closely-related bottlenecks at the intersection of number theory and sequence design: explicitly bounding the generalized quadratic Gauss sums and designing asymptotically order-optimal aperiodic polyphase sequence sets. %Our key contributions are twofold: 
	
We first derived an explicit upper bound for generalized quadratic Gauss sums (i.e., $|S_N(a/q, \theta)| < 20.07\sqrt{q} + 3$) through refined asymptotic expansion error analysis and by leveraging Fibonacci series convergence. Our derived upper bound provides a universal ``measurement tool" for analyzing sequences involving incomplete correlation/ambiguity sums. 
		
We then constructed four order-optimal aperiodic polyphase sequence sets with the aid of the aforementioned upper bound. Decimated Chu sequences were employed to construct $\mathcal{C}_1$ (for $K\geq 2$) with low aperiodic sidelobes and $\mathcal{C}_2$ with asymptotical LAZ properties. It is noted that $\mathcal{C}_1$ reduces to Mow's order-optimal Chu sequence pair for $K=2$. Besides, it was revealed that 1) the entire Alltop sequence set (i.e., $\mathcal{A}_1$)\footnote{Although it is known that Alltop sequences achieve the Welch bound on periodic cross-correlation, no further results on their aperiodic correlation properties have been reported before, to the best of our knowledge.} possesses asymptotically optimal low aperiodic correlation sidelobes; and 2) the Alltop sequence subset $\mathcal{A}_2$ achieves simultaneous asymptotic optimality with both low aperiodic correlation and ambiguity for the entire time-shift window.

As a future work, it would be interesting to understand if a tighter upper bound can be derived for the generalized quadratic Gauss sums and how it can be applied in other areas (e.g., physics, quantum information, spectral analysis). Secondly, new constructions of order-optimal aperiodic polyphase sequence sets that are not dependent on Chu or Alltop sequences\footnote{By taking advantage of the derived upper bound for generalized quadratic Gauss sums, it is possible for us to design other order-optimal aperiodic sequence sets arising from Chu or Alltop sequences. But it is more important to go beyond these two families of seed sequences.} are of interest to both the number theory and sequence design communities. Finally, the analysis and evaluation of these designed order-optimal aperiodic polyphase sequence sets in modern wireless systems (e.g., high mobility communication systems, and ISAC systems) are worthy of further study. 
		
\section{Acknowledgement}
Z. Liu would like to thank Professor Wai Ho Mow in the Hong Kong University of Science and Technology (HKUST) who introduced to him the order-optimal aperiodic polyphase sequence design problem during his visit to HKUST in June-July 2013. Some interesting observations through numerical computations were obtained but they got stuck due to the lack of relevant tools, until Z. Liu met H. Liu in January 2025.   

\end{document}